\documentclass{amsart}
\usepackage[english]{babel}
\usepackage{bbm}
\usepackage{hyperref}
\usepackage{verbatim}
\usepackage[pdftex]{graphicx}
\usepackage{epstopdf}
\usepackage{siunitx}

\usepackage{tikz}
\usepackage{pgf}
\usetikzlibrary{arrows,calc, positioning}
\usepackage{pgfplots}

\usepackage[super]{nth}

\newcommand{\N}{\ensuremath{\mathbb{N}}}
\newcommand{\R}{\ensuremath{\mathbb{R}}}
\newcommand{\Q}{\ensuremath{\mathbb{Q}}}
\newcommand{\Z}{\ensuremath{\mathbb{Z}}}
\newcommand{\E}{\ensuremath{\mathbb{E}}}
\renewcommand{\P}{\ensuremath{\mathbb{P}}}

\newcommand{\ind}[1]{\ensuremath{\mathbbm{1}_{\left\{#1\right\}}}}
\newcommand{\cal}[1]{\ensuremath{\mathcal{#1}}}

\def\ld{\lambda_2}
\def\lr{\lambda_3}
\def\sd{\sigma_2}
\def\sr{\sigma_3}
\def\var{\mathrm{var}}
\newcommand{\diff}{\mathop{}\mathopen{}\mathrm{d}}

\newtheorem{definition}{Definition}
\newtheorem{proposition}{Proposition}

\newtheorem*{notation}{Notations}

\newtheorem{theorem}{Theorem}

\title[Stochastic Gene Expression in Cells]{Stochastic Gene Expression in Cells: A Point Process Approach}
\date{\today}

\author[V. Fromion]{Vincent Fromion}
\email{Vincent.Fromion@jouy.inra.fr}
\address[V. Fromion, E. Leoncini]{MIG Mathematic, Computing Science and Genome, INRA,
Domaine de Vilvert, 78350 Jouy-en-Josas, France}
\author[E. Leoncini]{Emanuele Leoncini}
\email{Emanuele.Leoncini@inria.fr}
\author[Ph. Robert]{Philippe Robert}
\email{Philippe.Robert@inria.fr}
\urladdr{\href{http://www-rocq.inria.fr/~robert}{http://www-rocq.inria.fr/\~{}robert}}
\address[E. Leoncini,Ph. Robert]{INRIA Paris---Rocquencourt, Domaine de Voluceau, 78153 Le Chesnay, France}


\date{\today}

\begin{document}

\begin{abstract}
This paper investigates the stochastic fluctuations of the number of copies of a given protein in a cell.  This problem has already been addressed in the past and closed-form expressions of the mean and variance have been obtained for a simplified stochastic model of the  gene expression. These results have been obtained under the assumption that the duration of all the protein production steps are exponentially distributed. In such a case, a Markovian approach (via Fokker-Planck equations) is used to derive analytic formulas of the mean and the variance of the number of proteins at equilibrium.  This assumption is however  not totally satisfactory from a modeling point of view since the distribution of the duration of some steps is more likely to be Gaussian, if not almost deterministic. In such a setting,  Markovian methods  can no longer be used. A finer characterization of the fluctuations of the number of proteins is therefore of primary interest to understand the general economy of the cell. In this paper, we propose a new approach, based on marked Poisson point processes, which allows to remove the exponential assumption. This is applied in the framework of the classical three stages models of the literature:  transcription, translation and degradation. The interest of the method is shown by recovering the classical results under the assumptions that all the durations are exponentially distributed but also by deriving new analytic formulas when some of the distributions are not anymore exponential. Our results  show in particular that the exponential assumption may, surprisingly, underestimate significantly the variance of the number of proteins when some steps are in fact not exponentially distributed. This counter-intuitive result stresses the importance of the statistical assumptions in the protein production process. Finally, our approach can also be used to consider more detailed models of the gene expression.
\end{abstract}

\maketitle 

\bigskip

\hrule

\vspace{-3mm}

\tableofcontents

\vspace{-1cm}

\hrule

\bigskip

\section{Introduction}
The aim of the present work is to revisit and generalize the current mathematical results concerning the properties of intrinsic noise in \emph{gene expression}.  The stochastic characterisation of the gene expression in the protein production process has been theoretically studied by means of stochastic models in the late 70s by Berg~\cite{Berg1978} and Rigney~\cite{Rigney1979, Rigney1977} and reviewed recently by Paulsson~\cite{Paulsson}.  For a long period of time, it has not been possible to compare the theoretical results to real data, because of the lack of appropriate laboratory techniques. In the last two decades, the introduction of reliable expression reporter techniques and the use of fluorescent reporters, as the GFP (Green Fluorescent Protein), has allowed observations in live cells and the experimental quantification of the protein production at single cell level. See Taniguchi et al.~\cite{Taniguchi} for the experimental characterization of a large number of messengers and proteins of E. coli.

The good qualitative agreement observed between experiments and the predictions of these earlier stochastic models have stimulated further investigations to take into account the statistical characteristics of the phenomena involved in the protein production process. In this domain, the variance of the number of the cellular components in the cell is a key indicator of the efficiency of a production strategy, since it gives a measure of the fluctuation of resources of the cell consumed by the production process. Clearly enough, this characteristic is directly affected by model design and statistical  assumptions. 

As in the previous works, see Paulsson~\cite{Paulsson} and Swain~\cite{Swain2002}, it is especially important to derive explicit analytical formulas for the variance to assess the impact of the key parameters and, consequently, to get a biological interpretation of the obtained results. 

As it will be seen, the introduction of more realistic statistical assumptions leads to several technical difficulties, the main one being that the classic PDE approach (Fokker-Planck equations) used in the literature can no longer be used. In the present paper, using an approach based on marked Poisson point processes, we relax some statistical assumptions of the earlier stochastic models and obtain general results for a large class of stochastic models. In the biological context, we are then able to derive closed form expressions of the mean and variance of the main characteristics of the production process. 

\subsection{Biological Context and Model Motivations}
We first provide few biological insights about the protein production in living organisms.  The \emph{gene expression} is the process by which the genetic information is synthesised into a functional product, the proteins.  The production of proteins is the most important cellular activity, both for the functional role and the high associated cost in terms of resources (in prokaryotic cells it can reach up to $85$\% of the cellular resources). In particular, in a \emph{E. Coli} bacterium there are about $\num{3.6e6}$ proteins of approximately $2000$ different types with a large variability in concentration, depending on their types: from a few dozen up to $\num{e5}$.

The information flow from DNA genes to proteins is a fundamental process, common to all living organisms and is composed of two main elementary processes: \emph{transcription} and \emph{translation}.  During the transcription process, the \emph{RNA polymerase} binds to an active gene relative to a specific protein and makes a complementary copy of a specific DNA sequence, a \emph{messenger RNA} (mRNA). Each mRNA, which is a long chain of nucleotides, is a chemical ``blueprint'' for a particular protein.  The translation of the messenger into a polypeptide chain is achieved by a large complex molecule: the \emph{ribosome} with the help of some accessory factors like the elongation factor to cite a few.  During translation, the ribosome binds to the messenger and builds the polypeptide chain using mRNA as a template. More in detail, to each mRNA \emph{codon}, a triplet of nucleotides, corresponds a specific amino acid, which is the fundamental component of proteins.  The polypeptide chain of amino acids, folds spontaneously or with the help of \emph{chaperons}, into its functional three-dimensional structure. 

The gene expression is a highly stochastic process and results from the realization of a very large number of elementary stochastic processes of different nature.
The thermal excitation affects many processes, since it implies for example the free diffusion in the cytoplasm in which particles behave basically as if they were plunged into a viscous fluid.  In first approximation, three fundamental mechanisms are combined in the protein production.  The first is the pairing of two cellular components freely diffusing through the cytoplasm and is a direct consequence of the diffusion.  The second mechanism is the ``spontaneous'' rupture of the binding and the release of the two components as the result of thermal excitation.  The last main stochastic process involved is an active one, since it requires/uses energy, and corresponds  to the processing capability of both polymerase and ribosome. The active processes associated to polymerases and ribosomes are highly sophisticated steps, including for example dedicated proof reading mechanisms. In order to proceed to transcription initiation,  gene expression needs a successful binding of the polymerase to a specific DNA motif.  After the initiation step, the messenger chain is built through a series of specific stochastic processes, in which the polymerase recruits one of the four nucleotides in accordance to the DNA template.  A similar description is associated to the translation step.  In particular the protein elongation results in an iterative energy-consuming procedure in which each codon of the messenger chain is coupled with a particular tRNA, which adds a new amino-acid to the growing protein chain by means of ribosome.

In summary, most of the elementary processes  can be schematically seen as the encounter of two components in a viscous fluid. However, the classic approach to gene expression modeling is to group those elementary processes into critic steps as initiation, elongation and degradation, which are common to both transcription and translation. 

\subsection{Mathematical Model}
The corresponding mathematical model is now described.  For all the reasons given so far, the total number of copies of a given protein in the cell is a {\em random} variable $P$.  The cell can thus be thought as a system that produces a given protein with average concentration $\mathbb{E}[P]$, where $\mathbb{E}[X]$ denotes the expected value of a random variable $X$.  The protein concentration, which can vary of several orders of magnitude depending on the protein type, is in direct connection with the various parameters through a quite simple formula, as will be shown in the sequel.  The main objective of the paper is to derive an explicit representation of the variance of the number of proteins in terms of the various parameters of the protein production process.

\medskip
\noindent
\textbf{Gene activation}.
The gene activation involves complex processes among which the main ones are the association/dissociation of a repressor.

Usually the whole process is described as a telegraph process for which a transition from inactive state $0$ to active state $1$ occurs at rate $\lambda_1^+$ and, similarly from state $1$ to state $0$ at rate $\lambda_1^-$.  Here the fundamental assumption is that the distribution of these steps is exponential.  In a prokaryotic cell, there may be several copies of a specific gene and this fact has been included in few models in the past years, see Paulsson~\cite{Paulsson}.  Nevertheless, since we are interested in the variance of the number of proteins, we will assume in the following sections that there is only one copy of the gene.  The analogous result for the case with multiple copies is straightforward to obtain since, by independence, the variance of protein number is proportional to the number of copies of the gene.  

\medskip
\noindent
\textbf{Transcription}. 
A RNA polymerase binds on an active gene in an exponential time with rate $\lambda_2$. This effective rate measures the frequency of transcription initiation and takes into account several physical parameters, including, for example, the \emph{affinity} between the specific gene and the polymerase.
 The distribution $F_2$ on $\R_+$ of the lifetime $\sd$ of a mRNA is assumed to be general.

\medskip
\noindent
{\bf Translation}. Similarly, the binding of a ribosome on an mRNA occurs in an exponentially distributed time with rate $\lambda_3$, which measures the frequency of translation initiation and includes also the affinity between messenger and ribosome. The distribution $F_3$ of the lifetime $\sr$ of the protein is also general. The decay of the protein concentration occurs for two main reasons: by \emph{proteolysis}, \emph{i.e.} the protein degradation into amino acids, or by cellular dilution, due to the cellular volume increase of the bacterium during the exponential growth phase.

This paper is focused in the process of the production of a given protein. For this reason, the interaction with the production process of other proteins is not considered.  

\subsection{Literature: the three-stage model}
This is the fundamental model used to describe gene expression in the literature. 
We can already find these key steps in the first systematic and accurate studies of stochastic models for gene expression, as Rigney~\cite{Rigney1977, Rigney1979} and Berg~\cite{Berg1978}.
In recent years the three-stages model has been used as the fundamental structure in most well-known works of
Shahrezaei and Swain~\cite{Swain2008}, Paulsson~\cite{Paulsson} and Peccoud and Ycart~\cite{Peccoud}.

The promoter of the gene, corresponding to the specific protein of interest, can be in one of two possible states: active or inactive. In these studies transcription, translation and the degradation of proteins and messengers are modeled as first-order chemical reactions, \emph{i.e.} they are supposed to be exponentially distributed (or geometrically distributed in case of a discrete time setting). See Paulsson~\cite{Paulsson} for an extensive survey on the subject. With the above notations, this amounts to say that $\sd$ and $\sr$ are exponentially distributed. 

The assumption of exponentially distributed durations of the various phases of the three-stage model leads naturally to a Markovian modeling. The overall dynamic of gene activation can be described, see Paulsson~\cite{Paulsson},  by the random variable $Y(t)\in\{0,1\}$, where $Y(t)=1$ indicates that the gene is active at time $t$, while $Y(t)=0$ if it is inactive. Recall that we consider, without loss of generality, only the one gene case.  If we denote by $N_2(t)$ the number of mRNAs and by $N_3(t)$ the number of proteins, then it turns out that $(X(t))=(Y(t),N_2(t),N_3(t))$ is a Markov process with values in $\{0,1\}\times\N^2$.  This representation is common to most of the models of the literature. Some of them have, in fact, a lower dimensional state space because of assumptions on the number of mRNAs for example. As a consequence, the general theory of Markov processes gives a system of linear differential equations of order $1$, the Fokker-Planck equations, for the functions $p(t,(y,n_2,n_3))$, the probability that $X(t)$ is in state $(y,n_2,n_3)$ at time $t$. The system of equations has the general form
\begin{multline}\label{PDE}
\frac{\diff}{\diff t} p(t,(y,n_2,n_3))= \lambda_1(y) p(t,(1-y,n_2,n_3))+\lambda_2p(t,(y,n_2-1,n_3))\ind{y=1}\\+ \alpha(n_2) p(t,(y,n_2,n_3-1))+\beta(n_3) p(t,(y,n_2,n_3)).
\end{multline}
The solution of the system has a unique stable point $(\pi(y,n_2,n_3),(y,n_2,n_3)\in\{0,1\}\times\N^2)$, the invariant distribution of the Markov process, whose explicit expression is not known to the best of our knowledge.  Nevertheless, since the coefficients $\alpha(n)$ and $\beta(n)$ are linear with respect to $n$, the moments of the invariant distribution satisfy a recurrence equation. This equation is not completely simple, but gives an explicit expression for the first two moments and, in particular, for the variance, which is the key quantity to investigate these stochastic models.  This is the main theoretical result used in many papers in literature, see Rigney~\cite{Rigney1977}. It should be kept in mind that this approach is possible only under the assumption that {\em all} the duration of the main steps (like the production time of an mRNA or of a protein) are exponentially distributed. This assumptions is now discussed. 

\subsection{Statistical issues: the exponential assumption}
We refer to {\em exponential assumption} when the time to produce a particular cellular component and its lifetime, i.e. $\sd$ and $\sr$, are assumed to be exponentially distributed. 

The exponential assumption is natural in the following simple situation: if a large number of trials are necessary to achieve some goal (like the binding of some elements on a the DNA of an mRNA)  and each trial requires some duration $D$ and succeeds with probability $\alpha$. If $G_\alpha$ is the total number of attempts to succeed, i.e. $\P(G_\alpha\geq n)=(1-\alpha)^n$, then
\[
\lim_{\alpha\to 0} \P(\alpha G_\alpha\geq x)\sim e^{-x},
\]
in other words, if $\alpha$ is small then $\alpha G_\alpha\sim E_1$, where $E_1$ is an exponential random variable with mean $1$. Consequently, the total duration of time necessary to realize the objective is, due to the averaging of the law of large numbers ($G_\alpha$ is large), 
\[
\sum_{i=1}^{G_\alpha} D_i\sim G_\alpha \E(D)\sim \frac{\E(D)}{\alpha}E_1
\]
and is therefore exponentially distributed with mean $\E(D)/\alpha$. 

As it is seen, this scheme may describe correctly the duration of time to establish a binding of a polymerase or, of a ribosome. This scheme may properly describe the time required for a successful binding of RNA polymerase to the gene and of ribosome to mRNA. 

It should be noted that this assumption may not be true if one considers the elongation time of an mRNA or a protein chain. In particular during the polypeptide elongation, each tRNA, transporting a specific amino acid, should bind to the ribosome. If the distribution of the duration of this  step is indeed exponential, nevertheless the fact that elongation steps requires an average number of 100-300 steps, one for each amino acid, then the resulting distribution of the duration of the whole process is not anymore exponential. In first approximation, because of the large number of elongation steps, a deterministic elongation time with a small Gaussian perturbation should be considered. 
One of the main contributions of this paper is to show, via convenient mathematical tools, that the assumption on the distributions of $\sd$ and $\sr$ has an important impact on the qualitative properties of the protein production process.


\subsection{A Marked Point Process Description of Protein Production}
If the distributions of  $\sd$ and $\sr$ are not exponential, a Markovian description of the system is no longer possible, since the residual lifetimes of all the components have to be included in the state variable. In this case, to get a possible analogue of the PDE~\eqref{PDE}, an infinite dimensional state space would be required. For this reason, there is little hope to use, as it has been done up to now in the literature, the equivalent of Fokker-Planck equations to get explicit results like the first moments at equilibrium.  

Our approach consists in representing the state of the system as a functional of several marked point process. See the appendix for the general definitions and results concerning these processes.  If it may be difficult to have a PDE formulation to the problem, we can have a quite detailed description of the distribution of the number of proteins without solving recurrence equations by using an alternative method, which use some nice properties of the point processes.   See Robert~\cite{Robert}. The method is presented in the next section. An extension, see Fromion et al.~\cite{Fromion},  which uses the mathematical approach developped in this paper,  considers a  finer and more complete description of the gene expression. In particular it includes the dilution process during the exponential growth phase. 

\subsection{Outline of the Paper}
Section~\ref{secmodel} introduces the marked Poisson point processes used in the mathematical modeling of the production of proteins. In this model the lifetime of an mRNA or of a protein has a general distribution instead of the exponential assumption used in the models in the literature. Appendix~\ref{secapp} recalls briefly the main results concerning this class of point processes. Section~\ref{secmRNA} gives the main results concerning the equilibrium distribution of the number of mRNAs at equilibrium, the main tool in this analysis is the representation in terms of marked Poisson point processes and a coupling argument.  Section~\ref{secProt} is devoted to the derivation of an explicit formula for the variance of the number of proteins. Several examples of distributions are discussed. 

\section{Stochastic Model}\label{secmodel}
In this section, the various stochastic processes are introduced. In  the appendix we recall the main results and notations concerning the marked Poisson point processes (MPPP) which are used in this paper. 
\subsection*{Gene activation} 
It is assumed that there is one active gene, which is activated at rate $\lambda_1^+$ and inactivated at rate $\lambda_1^-$. Recall  that the assumption that $n_{\text{max}}$ the maximum number of active genes is $1$ does not restrict the generality of our results since the quantities analyzed in this paper (expected values and variances) are proportional to $n_{\text{max}}$.  Let $(E_n)$ and $(F_n)$ be i.i.d. exponential random variables with respective rates $\lambda_1^-$ and $\lambda_1^+$. The process of activation of the gene at equilibrium can be represented as a stationary process $(Y(t), t\in\R)$ with values in $\{0,1\}$. Note that $(Y(t))$ is defined on the whole real line, i.e. that the activation/deactivation process has started at $t=-\infty$. As it will be seen, this is a convenient representation to describe properly the equilibrium of the protein production process. 
The increasing sequence of the instants of activation of the gene is denoted by $(t_n)$ with the convention that $t_0\leq 0<t_1$. In particular
\[
\{t_n, n\in\Z\}=\{s\in\R: Y(s-)=0 \text{ and } Y(s)=1\}
\]
and $t_{n+1}-t_n=E_n+F_n$. Because of our assumption $(t_n)$ is a stationary renewal point process. 




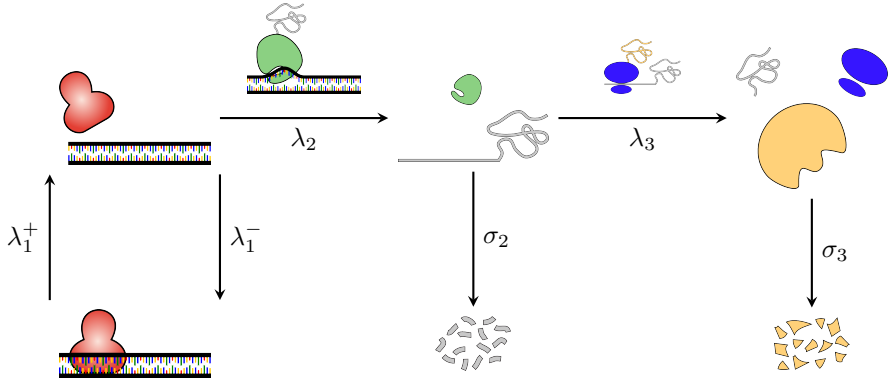
\begin{figure}[h]
	\begin{center}
 	  	\begin{tikzpicture}[node distance=4.5cm,->,>=stealth,thick]
  		\node  (active)  {\includegraphics[width=2cm]{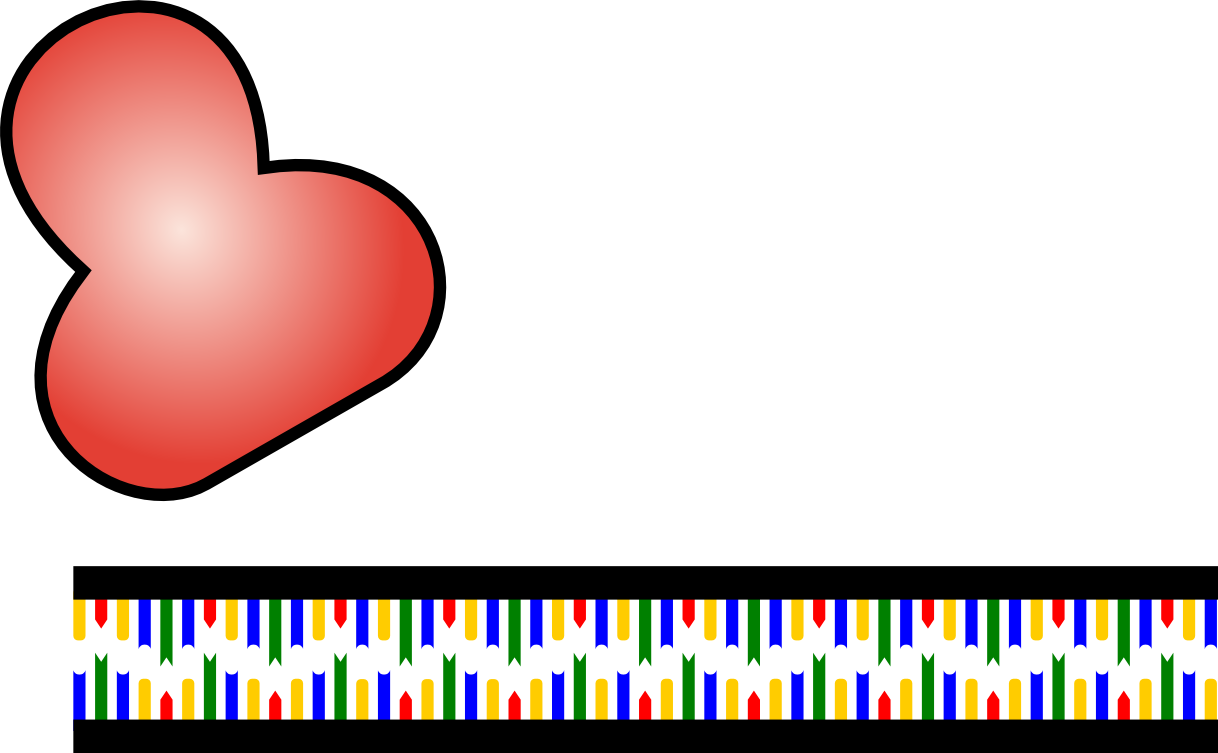}};
		\node (inactive) [below of=active, node distance = 3cm]   {\includegraphics[width=2cm]{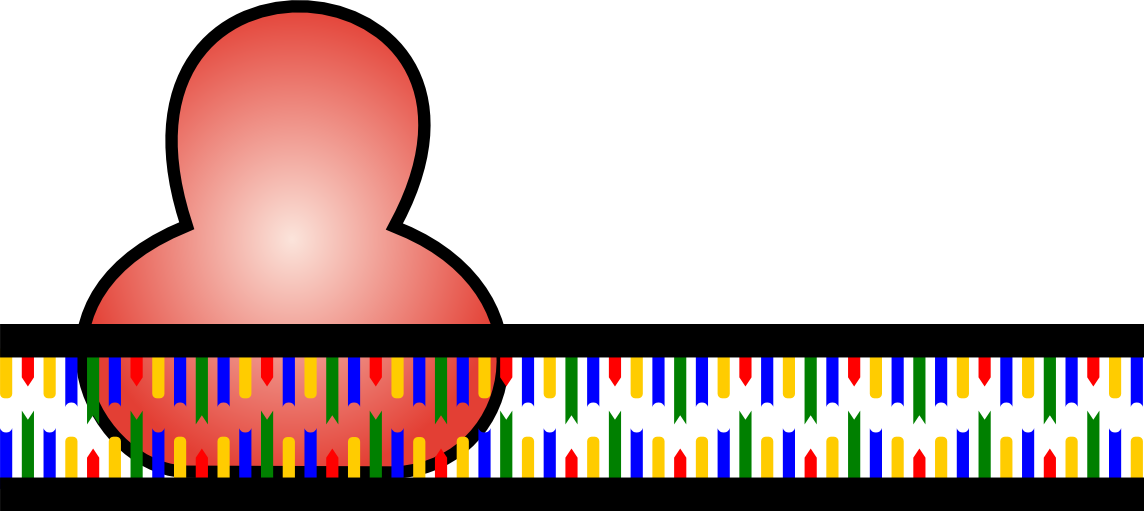}};
		\node (RNA) [right of=active] {\includegraphics[width=2cm]{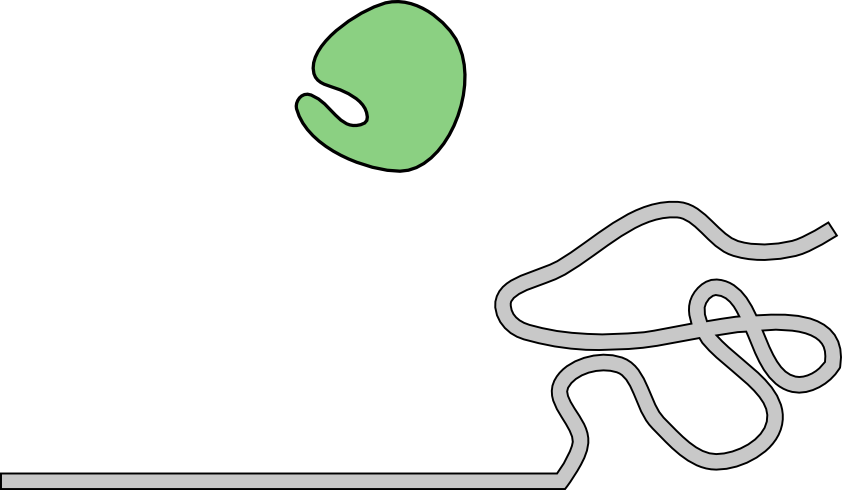}};
		\node (deadRNA) [below of=RNA, node distance = 3cm] {\includegraphics[width=1cm]{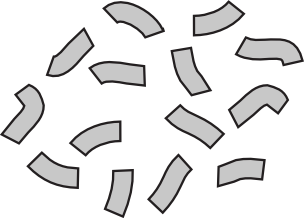}};
		\node (Prot) [right of=RNA] {\includegraphics[width=2cm]{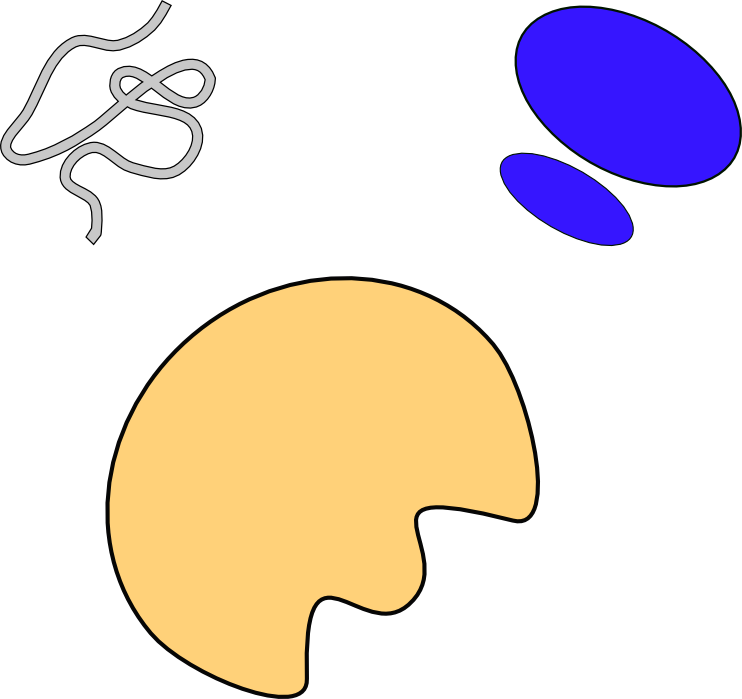}};
		\node (deadProt) [below of=Prot, node distance = 3cm] {\includegraphics[width=1cm]{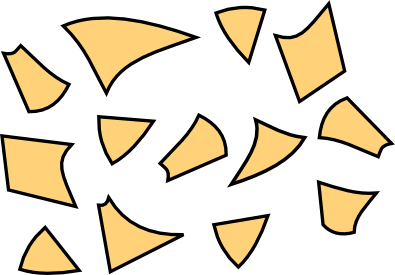}};
				
		\path	(active.south east) edge node [right] {$\lambda_1^-$} (inactive.north east)
						(inactive.north west) edge node [left] {$\lambda_1^+$} (active.south west);
		\path	(RNA) edge node [right] {$\sigma_2$} (deadRNA);
		\path	(Prot) edge node [right] {$\sigma_3$} (deadProt);
		
		\path (active) edge node [below] {$\lambda_2$} node[above=0.2cm]{\includegraphics[width=1.5cm]{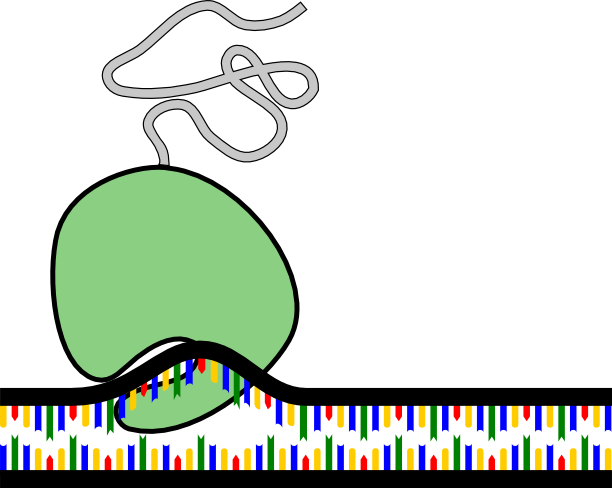}}  (RNA);
		\path (RNA) edge node [below] {$\lambda_3$} node[above=0.2cm]{\includegraphics[width=1cm]{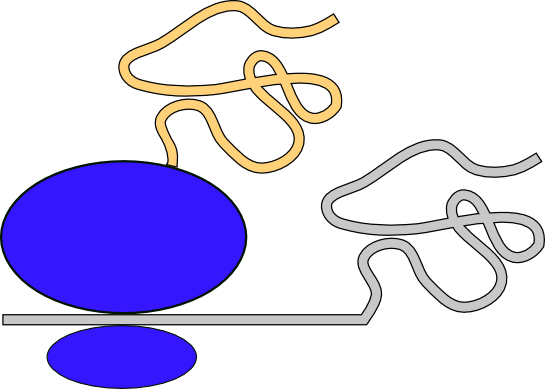}} (Prot);
  	\end{tikzpicture}
  	\end{center}
  	 \caption{Three stage model. The gene activation/deactivation occur at rate $\lambda_1$ and $\mu_1$ respectively. Transcription
  and translation occur at rates $\lambda_2$ and $\lambda_3$ respectively. The degradation times of mRNAs and proteins have probability
  distributions $F_2(\diff t)$ and $F_3(\diff t)$ respectively.}
  \label{fig:fourStageModel}
\end{figure} 

\subsection*{Production of mRNAs}
When the gene is active, it produces mRNAs at rate $\ld$ and $F_2(\diff y)$ is the
distribution of the lifetime of a mRNA. Let  ${\cal N}_{\ld}=(s_n,{\sd}_{,n})$ be a
MPPP on $\R_+^2$ with intensity measure $\ld \diff x\otimes F_2(\diff y)$.

If the gene is, for $s\leq t$, then the formula
\[
{\cal N}_{\ld}([s,t]\times\R)=\sum_{n\in\Z}  \ind{s\leq s_n\leq t}=\int \ind{s\leq u\leq t} {\cal N}_{\ld}(\diff u,\diff v)
\]
represents the total number of mRNAs created between time $s$ and time $t$ and
\[
\sum_{n\in\Z}  \ind{s\leq s_n\leq t \leq s_n+{\sd}_{,n}}=\int \ind{s \leq u\leq t\leq u+v} {\cal N}_{\ld}(\diff u,\diff v)
\]
is the number of mRNAs still alive at time $t$. 
More in general, if we include the gene dynamics into the formula, we find that the number of messengers created in the time interval $[s,t]$
and still alive at time $t$ is
\[
\sum_{n\in\Z}  \ind{s\leq s_n\leq t \leq s_n+{\sd}_{,n},Y(s_n)=1}=\int \ind{s\leq u\leq t\leq u+v, Y(u)=1} {\cal N}_{\ld}(\diff u,\diff v).
\]

\subsection*{Production of Proteins}
A given mRNA produces proteins at rate $\lr$ and $F_3(\diff y)$ is the distribution of the duration of the lifetime of a protein. 

For $u\in\R$, denote by ${\cal N}_{\lr}^{u}$ a MPPP with intensity $\lr \diff x\otimes F_3(\diff y)$. In the following it is the  process of creation of proteins associated to an mRNAs created at time $u$. In particular, if mRNA lifetime is $v$ then
\[
{\cal N}_{\lr}^{u}([u,u+v]\times\R_+)= \int_{[u,u+v]\times\R_+} \,{\cal N}_{\lr}^{u}(\diff x,\diff y)
\]
is the total number of proteins created by such an mRNA during its lifetime.

\bigskip
\noindent
{\bf Remarks.}\\
Here the mRNA is available for translation once a small portion of the growing mRNA chain has been assembled. This assumption is coherent with the prokaryotic dynamics, but should be adapted for the eukaryotic case. In fact in this case we have to wait the completed messenger to be exported to the cytoplasm. If we assume this time to be deterministic, then the previously defined integral should be shifted of a constant value and we should easily get the corresponding analytic results.

The whole process of production of mRNAs and proteins can thus be described by the sequence
\[
{\cal A}=\left(s_n,t_n,{\cal N}_{\lr}^{s_n}\right).
\]
Recall that ${\cal N}_{\lr}^{0}: (\Omega, \cal{F}, \cal{P}) \to {\cal M}_{p}(\R\times\R_+)$, where ${\cal M}_{p}(\R\times\R_+)$ is the set of point processes on $\R\times\R_+$.
If we denote with $\Q$ the distribution of ${\cal N}_{\lr}^0$ on ${\cal M}_{p}(\R\times\R_+)$, 
the process ${\cal A}$ can be seen as a marked Poisson point process on $\R_+\times {\cal M}_{p}(\R\times\R_+)$ with intensity measure $F_3(\diff x)\times\Q$. 
This observation will not be used in the following to keep the setting as simple as possible but the proof of Proposition~\ref{varPprop} below could be shortened by using it together with Proposition~\ref{remind}. 

The notations with some definitions for the stochastic models used in this paper are now summarized.
\begin{notation}
\begin{itemize}\rm
\item {\em Gene activation.}\\
The activation rate is [resp. inactivation rate]  is $\lambda_1^+$ [resp.  $\lambda_1^-$] and
\[
\delta_+=\frac{\lambda_1^+}{\Lambda} \text{ and } \Lambda=\lambda_1^++\lambda_1^-.
\]
\item {\em mRNA production.}\\
The rate of production of mRNAs by an active gene is $\ld$, $F_2(\diff x)$ is the distribution of an mRNA lifetime, $\sd$ denotes a random variable with distribution $F_2$ and
\[
\rho_2\stackrel{\text{def.}}{=} \ld \E(\sd)=\ld \int_{\R_+}x F_2(\diff x). 
\]
\item {\em Protein production.}\\
The rate of production of proteins by an mRNA is $\lr$, the lifetime distribution of a protein  is $F_3(\diff x)$, $\sr$  denotes a random variable with distribution $F_3$ and
\[
\rho_3\stackrel{\text{def.}}{=} \lr \E(\sr)=\lr \int_{\R_+}x F_3(\diff x). 
\]

\end{itemize}
\end{notation}

\section{Equilibrium Distribution of the Number of  mRNAs}\label{secmRNA}
This section investigates the first part of the protein production process: activation of the convenient gene and production of mRNAs. 
\subsection{State of the Gene}
The behavior of the process $(Y(t))$ is well known. Once equilibrium has been reached, it results
\[
\P(Y(0)=1)=\delta_+=\frac{\lambda_1^+}{\lambda_1^++\lambda_1^-}=1-\P(Y(0)=0). 
\]
To express the variance of the number of proteins, the following quantity is required, for $t\geq 0$, 
\begin{equation}\label{covY}
\P(Y(t)=1|Y(0)=1)=\delta_++(1-\delta_+)e^{-\Lambda t},
\end{equation}
with $\Lambda=\lambda_1^++\lambda_1^-$. See Norris~\cite{Norris} and Peccoud and Ycart~\cite{Peccoud} for detailed computations.
From now on, it will be assumed that $(Y(t))$ is defined on $\R$ and is at equilibrium. 

\subsection{Number of mRNAs}
A result on the number of mRNAs at equilibrium and its distribution is derived in this section. The techniques used to prove it will also be used to investigate the distribution of the number of proteins in the next section.
In order to present the MPPP approach, we will develop computations for mRNAs, since they are simpler from the point of view of notations, but include the main ideas.

\begin{proposition}\label{mRNAeq}
The number $M$ of mRNA's at equilibrium  can be represented as
\begin{equation}
M= \int_{\R\times\R_+}  \ind{u\leq 0\leq u+v, Y(u)=1}  \,{\cal N}_{\ld}(\diff u,\diff v),
\end{equation}
where ${\cal N}_{\ld}$ is a Poisson marked point process with intensity  $\ld \diff x\otimes F_2(\diff y)$.
\end{proposition}
\begin{proof}
Suppose there are no mRNAs at starting time $0$,
then the number $M_t$ of mRNAs at time $t$ is given by
\[
M_t=\sum_{n} \ind{0\leq s_n\leq t \leq s_n+{\sd}_{,n}, Y(s_n)=1}= \int_{\R_+} \int_0^t  \ind{u\leq t\leq u+v, Y(u)=1}  \,{\cal N}_{\ld}(\diff u,\diff v),
\]
if ${\cal N}_{\ld}=(s_n,{\sd}_{,n})$ as defined in Section~\ref{secmodel}. Recall that $s_n$ is the (potential) $n$th binding time of a polymerase on the gene: an mRNA is created only if the gene is active, i.e. $Y(s_n)=1$. The term ${\sd}_{,n}$ represents the lifetime of the newly produced mRNA.
The right-hand-side of the previous equation accounts for the number of mRNAs produced in the interval $[0,t]$ and still alive at time $t$ ($u+v\geq t$). 

Since the process $(Y(t))$ is stationary as well as the Poisson marked point process, they are both invariant by translation. By translating by $-t$, one gets that $M_t$ has the same distribution as 
\[
M_t\stackrel{\text{dist.}}{=} \int_{\R_+} \int_{-t}^0  \ind{0\leq u+v, Y(u)=1}  \,{\cal N}_{\ld}(\diff u,\diff v),
\]
by letting $t$ go to infinity, one obtains the desired result. 
\end{proof}

\noindent
{\bf Remark}\\ It is crucial that the distribution of $M_t$ can be explicitly expressed as a functional of the marked Poisson process ${\cal N}_{\ld}$. The same property is true for its limit.  In this context, with the help of the coupling argument, there is no need of a Markovian setting to prove that $M_t$ converges in distribution as $t$ goes to infinity. As will be seen, the distribution of the limit $M$ can be obtained by using some properties of Poisson point processes.  For all these reasons, there is no need to impose the random variables $\sd$ and $\sr$ to be exponentially distributed.

In the proof of the above result, we have in fact proved a more general result.
\begin{theorem}\label{thmRNA}
 The point process ${\cal M}$ representing the instants of creation of mRNAs and the associated
lifetime at equilibrium can be represented as
\begin{equation}\label{mRNApp}
{\cal M}= \int_{\R \times \R_+}  \ind{Y(u)=1} \delta_{(u,v)} \,{\cal N}_{\ld}(\diff u,\diff v),
\end{equation}
where $\delta_z$ is the Dirac mass at $z$.
\end{theorem}

The number of mRNAs alive at equilibrium can thus be represented as
\[
M =\int \ind{ u\leq t\leq u+v} {\cal M}(\diff u,\diff v)=\int \ind{ u\leq t \leq u+v, Y(u)=1} {\cal N}_{\ld}(\diff u,\diff v)
\]
which is precisely the expression of Proposition~\ref{mRNAeq}. When the activation rate of the gene goes to infinity, the point process ${\cal M}$ is simply a marked Poisson point process and $M$ has a Poisson distribution with parameter $\rho_2=\ld\E(\sd)$. 

We now use this representation to get an explicit expression of the variance of the number of mRNAs at equilibrium. 
\begin{proposition}\label{mRNAprop}
If the distribution of the lifetime of a mRNA is $F_2(\diff x)$, the average of the number $M$ of mRNAs at equilibrium is given by 
\[
\E(M)=\delta_+ \rho_2=\frac{\lambda_1^+}{\lambda_1^++\lambda_1^-} \ld\int xF_2(\diff x).
\]
The variance of $M$ is
\begin{equation}\label{varmRNA}
\var(M)=\E(M)
+2\rho_2^2\delta_+(1-\delta_+)\int_0^{+\infty} e^{-\Lambda v} \overline{F}_2(u)\overline{F}_2(u+v)\,\diff u\, \diff v
\end{equation}
where $F_2(x)=F_2([0,x])$ and $\overline{F}_2(x)=(1-F_2(x))/\E(\sd)$. 
\end{proposition}
\begin{proof}
Conditionally on the process $(Y(t))$, $M$ follows a Poisson distribution, hence for $z\in[0,1]$,
\begin{align}
\E\left(z^M\mid (Y(t))\right)
&=\exp\left({-}\ld(1{-}z) \int_{\R_+} \int_{-\infty}^{0}  \ind{Y(u)=1,u+v>0}\,\diff u\,\sd(\diff v)\right) \notag \\
&=\exp\left({-}\ld(1{-}z) \int_0^{+\infty}  \ind{Y(-u)=1}\P(\sd\geq u) \,\diff u\right)
\label{eq1}
\end{align}
by taking $f(u,v)=\ind{Y(u)=1,u\leq 0,u+v>0}$ in Relation~\eqref{Lap}.
If we differentiate formula \eqref{eq1} with respect to $z$ and take $z=1$, we obtain
\[
\E\left(M\mid (Y(t))\right)=
\ld \int_{-\infty}^{0} \ind{Y(u)=1}\P(\sd\geq -u) \,\diff u,
\]
since $(Y(t))$ is at equilibrium,  $\P(Y(u)=1)=\delta_+$, hence integrating the last relation we get
\[
\E(M)=
\delta_+\ld \int_{-\infty}^{0} \P(\sd\geq -u) \,\diff u
=\delta_+ \ld\E(\sd).
\]
If we differentiate twice Formula~\eqref{eq1} and substitute $z=1$, we obtain
\begin{multline*}
\E(M(M-1)\mid (Y(t)))=\ld^2\left(\int_{0}^{+\infty}\ind{Y(-u)=1}\P(\sd\geq u) \,\diff u\right)^2
\\=\ld^2\int_{\R_+^2} \ind{Y(-u)=1,Y(-v)=1}\P(\sd\geq u) \P(\sd\geq v) \,\diff u\,\diff v,
\end{multline*}
which, integrated with respect to $(Y(t))$, gives
\[
\E\left(M^2\right)-\E(M)=\ld^2\int_{\R_+^2} \P(Y(-u)=1,Y(-v)=1)\P(\sd\geq u, \overline{\sd}\geq v) \,\diff u\,\diff v,
\]
where the random variable $\overline{\sd}$ is independent of $\sd$ and has the same distribution.  
Using relation~\eqref{covY}, for $u\leq v$ and $\Lambda=\lambda_1^++\lambda_1^-$, we get
\begin{multline*}
\P(Y(-u)=1,Y(-v)=1)=\P(Y(-v)=1)\P(Y(-u)=1\mid Y(-v)=1)
\\=\delta_+\left(\delta_++ (1-\delta_+)e^{-\Lambda(v-u)}\right).
\end{multline*}
Therefore $\E(M^2)-\E(M)$ is the sum of 
\[
\ld^2\delta_+^2 \int_{\R_+^2} \P(\sd\geq u, \overline{\sd}\geq v) \,\diff u\,\diff v=(\ld\delta_+\E(\sd))^2=(\E(M))^2
\]
and, up to the multiplicative factor $2\ld^2\delta_+(1-\delta_+)$, of
\[
\int_{\R_+^2} \P(\sd\geq u, \overline{\sd}\geq v)e^{-\Lambda(v-u)}\ind{u\leq v} \,\diff u\,\diff v.
\]
The proposition is proved. 
\end{proof}

\subsection*{Normalized variance}
By Relation~\eqref{varmRNA}, the normalized variance of $M$ is defined as 
\[
\frac{\var(M)}{\E(M)^2}=\frac{1}{\E(M)}+2\frac{1-\delta_+}{\delta_+}\int_0^{+\infty} e^{-\Lambda v} \overline{F}_2(u)\overline{F}_2(u+v)\,\diff u\, \diff v.
\]
When the mean $\E(M)$ is fixed, the only quantity  which depends on the distribution of the lifetime of an mRNA is the integral
\[
I_{F_2}= \int_0^{+\infty} e^{-\Lambda v} \overline{F}_2(u)\overline{F}_2(u+v)\,\diff u\, \diff v.
\]
To conclude this section, 
we now apply the previous general formulas to specific choices of the probability distribution. In particular we will get analytical formula of the previous integral for exponential and deterministic distributions. These assumptions are not completely realistic from a biologic point of view, nevertheless they are used to stress the impact of probability distribution on the messenger variance.
If the distribution of the lifetime of an mRNA  is the exponential distribution $E_{\mu_2}$  with parameter $\mu_2$, one gets 
\[
I_{E_{\mu_2}}=\frac{1}{2\mu_2(\Lambda+\mu_2)}.
\]
If the lifetime of an mRNA  is the deterministic distribution $D_{\mu_2}$ with a unit mass at  $1/\mu_2$, the above formula yields
\[
I_{D_{\mu_2}}=\frac{1}{\Lambda^2}\left(e^{-\Lambda/\mu_2}-1+\frac{\Lambda}{\mu_2}\right).
\]
Straightforward calculations with these formulas show that $I_{E_{\mu_2}}{\leq}I_{D_{\mu_2}}$. The ratio $I_{D_{\mu_2}}/I_{E_{\mu_2}}$ varies in fact between $1$ and $2$, see Figure~\ref{fig0}.  The variance for the exponential distribution is smaller than the one for the deterministic distribution with the same mean. This result is not quite intuitive if one takes into account that the variance of the exponential distribution is quite large. 

\begin{figure}[htp]\label{fig0}
\begin{center}
\scalebox{0.8}{\includegraphics{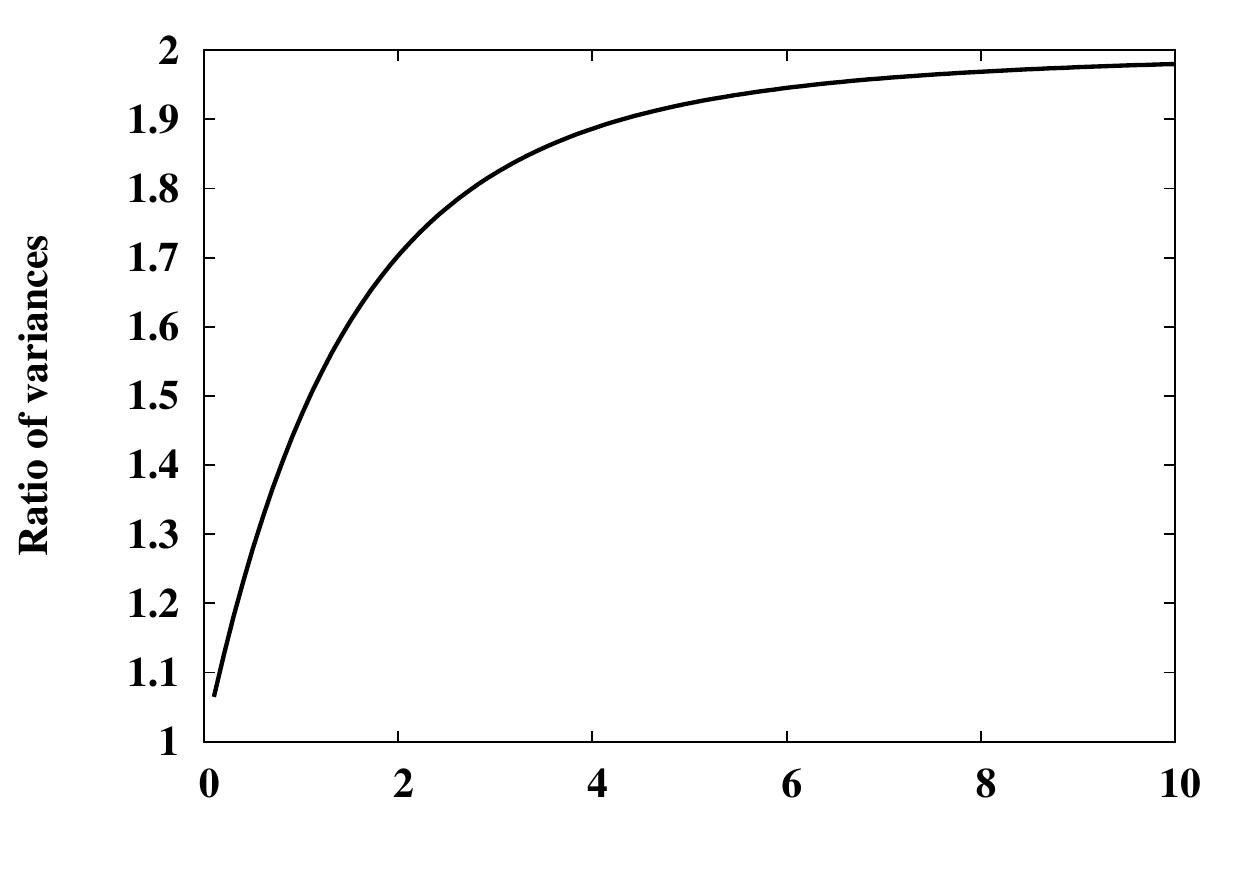}}
\put(-285,155){\rotatebox{90}{$\mathbf{I_{D_{\mu_2}}/I_{E_{\mu_2}}}$}}
\put(-55,5){$\mathbf{\Lambda/\mu_2}$}
\end{center}
\caption{Ratio of Variances of nb of mRNAs: Deterministic/Exponential}
\end{figure}

\section{Variance of the Number of  Proteins at equilibrium}\label{secProt}
Recall that if an mRNA is created at time $u$ and has a lifetime $v$, then on the time interval $[u,u+v]$ proteins are created according to the marked Poisson point process ${\cal N}_{\lr}^{u}$ with intensity $\lr\,\diff x\otimes\, F_3(\diff y)$.
 The instants of creation of proteins together with their lifetimes can thus be represented by the following point process
\begin{equation}\label{protpp}
{\cal P}=\int_{\R\times\R_+}{\cal M}(\diff u,\diff v) \int_{[u,u+v]\times\R_+} \delta_{(x,y)}\,{\cal N}_{\lr}^{u}(\diff x,\diff y),
\end{equation}
where ${\cal M}$ is the point process defined by formula~\eqref{mRNApp}.

\begin{proposition}
The number $P$ of proteins at equilibrium can be represented by the random variable
\begin{equation}\label{protnb}
P= \int_{\R\times \R_+} \ind{Y(u)=1}\, {\cal N}_{\ld}(\diff u,\diff v)\int_{\R\times\R_+}  \ind{x\leq 0 \leq x+y, u\leq x\leq u+v } {\cal N}_{\lr}^{u}(\diff x,\diff y).
\end{equation}
\end{proposition}
\begin{proof}
The derivation is quite straightforward. If an mRNA alive between time $u$ and $u+v$ generates a protein at time $x$ with lifetime $y$, this protein will be present at time $0$ if $x\leq 0\leq x+y$. The argument that this is indeed the representation of the number of proteins at equilibrium follows the same lines of the proof of Proposition~\ref{mRNAeq}.
\end{proof}

Before the main technical result of the paper, we can get information on the distribution of $P$ using formula \eqref{protnb}.
We start with the simple case of the mean. For fixed $u$, $v\in\R_+$, formula \eqref{meanpp} gives
\[
\E\left(\int_{\R\times\R_+}  \ind{\substack{x\leq 0 \leq x+y,\\ u\leq x\leq u+v} } {\cal N}_{\lr}^{u}(\diff x,\diff y) \right)
=\lr\int_{\R\times\R_+}  \ind{\substack{x\leq 0 \leq x+y,\\ u\leq x\leq u+v }} \,\diff x\,F_3(\diff y) 
\]
Integrating this expression with respect to $\ind{Y(u)=1}\, {\cal N}_{\ld}(\diff u,\diff v)$ and taking its expectation, we get
\begin{eqnarray*}
    \lefteqn{\E [P\mid (Y(t))] =} \\
    &=& \lr\left.\E\left(\int \ind{Y(u)=1}\left[ \int \ind{\substack{x\leq 0 \leq x+y,\\ u\leq x\leq u+v }} \,\diff x\,\sr(\diff y) \right] \, {\cal N}_{\ld}(\diff u,\diff v)\right|(Y(t))\right)\\
	&=& \lr\int \ind{Y(u)=1}\left[\int \ind{\substack{x\leq 0 \leq x+y,\\ u\leq x\leq u+v }} \,\diff x\,\sr(\diff y) \right]\,\ld \diff u\,F_2(\diff v)\\
	&=& \ld\lr\int \ind{Y(u+x)=1} \P(\sd\geq -u)\P(\sr\geq -x)\,\diff x\,\diff u,
\end{eqnarray*}
where we used again formula \eqref{meanpp}. A further integration gives finally the expectation
\begin{align*}
\E(P)&=\ld\lr\int_{\R} \P(Y(u+x)=1)\P(\sd\geq -u)\P(\sr\geq -x)\,\diff x\,\diff u\\
&=\ld\lr\delta_+\int \P(\sd\geq -u)\P(\sr\geq -x)\,\diff x\,\diff u=\delta_+\ld\E(\sd)\lr\E(\sr),
\end{align*}
with the notation introduced in Proposition~\ref{mRNAprop}.

\begin{theorem}\label{varPprop}
If the distribution of the lifetime of a mRNA [resp. protein] is $F_2(\diff x)$ [resp. $F_3(\diff y)$], then 
the expected value of the random variable $P$, which is the number of proteins at equilibrium, is given by
\[
\E(P)=\delta_+\rho_2\rho_3=\frac{\lambda_1^+}{\lambda_1^++\lambda_1^-} \ld\int xF_2(\diff x)\lr\int xF_3(\diff y)
\]
and its variance $\var(P)$ can be expressed as 
\begin{align}\label{varPeq}
\var(P) &= \E(P) + \lambda _2 \rho_3^2 \delta_+ \int _0^{+\infty} \int_{\mathbb{R}_+} \left[
\int_{-s}^{(-s+t)\wedge 0}\overline{F}_3(u)\,\diff u \right]^2 \,\diff s F_2(\diff t)  \\
 &+ \rho_2^2 \rho_3^2 \delta_+(1-\delta_+) \int _{\mathbb{R} ^4 _+}
e^{-\Lambda|(u_1-u_2)+(v_1-v_2)|} \prod_{i=1}^2\overline{F}_2(u_i)\overline{F}_3(v_i) \diff u_i\diff v_i,\notag
\end{align}
where, for $j=2$, $3$, $F_j(x)=F_j([0,x])$ and $\overline{F}_j(x)=(1-F_j(x))/\E(\sigma_j)$. 
\end{theorem}
\begin{proof}
Recall that ${\cal N}_{\ld}$ can also be represented as ${\cal N}_{\ld}=(s_n,t_n)$ and
\[
P= \sum_{n\in\Z}  \int_{\R\times\R_+}  \ind{Y(s_n)=1}\ind{\substack{x\leq 0 \leq x+y,\\ s_n\leq x\leq s_n+t_n} } {\cal N}_{\lr}^{s_n}(\diff x,\diff y).
\]
Denote by $\widehat{\E}$ the conditional expectation $\E(\cdot \mid (Y(t)),\, (s_n,t_n))$. 
The conditional generating function $\widehat{\E}\left(z^P\right)$ can be written as
\begin{multline*}
\widehat{\E}\left(\prod_{n\in \Z} \exp\left(-\log(z)\int_{\R\times\R_+}  \ind{Y(s_n)=1}\ind{\substack{x\leq 0 \leq x+y,\\ s_n\leq x\leq s_n+t_n} } {\cal N}_{\lr}^{s_n}(\diff x,\diff y)\right)\right)
\\=\prod_{n\in \Z}\widehat{\E}\left( \exp\left(-\log(z)\int_{\R\times\R_+}  \ind{Y(s_n)=1}\ind{\substack{x\leq 0 \leq x+y,\\ s_n\leq x\leq s_n+t_n }} {\cal N}_{\lr}^{s_n}(\diff x,\diff y)\right)\right),
\end{multline*}
since the point processes ${\cal N}_{\lr}^{s_n}$, $n\in\Z$, are independent.

The $n$th term of this product is, applying Proposition~\ref{remind} to the marked Poisson point processes ${\cal N}_{\lr}^{s_n}$,
\begin{multline*}
\exp\left(-\lr(1-z)\ind{Y(s_n)=1}\int \ind{\substack{x\leq 0 \leq x+y,\\ s_n\leq x\leq s_n+t_n} } \,\diff x\,F_3(\diff y)\right)
\end{multline*}
By integrating $\widehat{\E}\left(z^P\right)$ with respect to ${\cal N}_{\ld}$, the generating function  can thus be written as 
\[
\E\left(z^P|(Y(t))\right)=\E\left(\exp\left(-\int g(u,v){\cal N}_{\ld}(\diff u,\diff v)\right)\right),
\]
where
\[
g(u,v)=\lr(1-z)\ind{Y(u)=1}\int \ind{\substack{x\leq 0 \leq x+y,\\ u\leq x\leq u+v} } \,\diff x\,F_3(\diff y).
\]
Applying again  Proposition~\ref{remind} to the marked Poisson point process ${\cal N}_{\ld}$, we get 
\begin{eqnarray*}
    \lefteqn{\E\left(z^P|(Y(t))\right) =} \\
&=& \exp\left(-\ld\int_{\R} \diff u\int_{\R} F_2(\diff v)\left(1{-}\exp\left(\rule{0mm}{4mm}{-}\lr(1-z)\int \ind{\substack{x\leq 0 \leq x+y,\\ u\leq x\leq u+v\\Y(u)=1} } \,\diff x\,F_3(\diff y)\right)\right)\right).
\end{eqnarray*}
In order to obtain an expression for $\E\left(P(P-1)|(Y(t))\right)$, we have to differentiate twice the previous formula with respect to $z$ and evaluate it at $z=1$.
The resulting formula should then be integrated with respect to $(Y(t))$ and we can get formula \eqref{varPeq}, by using similar arguments as in the proof of Proposition~\ref{mRNAprop} (with more technical calculations). 
\end{proof}

\bigskip
\noindent
{\bf Applications.}\\
To show the effectiveness of the analytic formula~\eqref{varPeq} of the protein variance, one considers the cases of exponential and deterministic distributions. More realistic cases are considered, see the figure. This specific analysis will give an indication of the impact of the distribution on the protein variance. In each case the average lifetime of an mRNA [resp. protein] is $1/\mu_2$ [resp. $1/\mu_3$]. Recall that $\delta_+=\lambda_1^+/\Lambda$ and $\Lambda=\lambda_1^++\lambda_1^-$. As in the case of mRNAs above, if from a biological point of view these assumptions are not completely realistic, this analysis shows the impact of the distribution on the variance, and therefore of the necessity of having closed form expressions for a large set of distributions. 

\bigskip
\noindent
{\bf Exponential Distribution.}\\
If the distribution of the lifetime of an mRNA [resp. protein] is exponential  with parameter $\mu_2$ [resp. $\mu_3$],  then formula \eqref{varPeq} gives the classical result on the variance, see Paulsson~\cite{Paulsson}, 
\begin{equation}
\var_E(P)= \E(P)
\left(1+\frac{\lambda_3}{\mu_2+\mu_3}+\frac{\lambda_2\lambda_3(1-\delta_+)(\Lambda+\mu_2+\mu_3)}{(\mu_2+\mu_3)(\Lambda+\mu_2)(\Lambda+\mu_3)}\right).
\end{equation}

\bigskip
\noindent
{\bf Deterministic Case.}\\
If the lifetime of an mRNA is exponentially distributed with parameter $\mu_2$ and the protein lifetime is deterministic, equal to $1/\mu_3$, then formula \eqref{varPeq} gives the identity
\begin{multline}
\var_D(P)= \E(P)
\left[
\rule{0mm}{6mm}
1+2\frac{\lambda_3}{\mu_2}\left(1-\frac{\mu_3}{\mu_2}\left(1-e^{-\mu_2/\mu_3}\right)\right)
\right.\\
+\frac{2\lambda_2\lambda_3(1-\delta_+)\mu_2}{\Lambda^2-\mu_2^2}
\left(
\frac{\mu_3}{\Lambda^2}\left[1-e^{-\Lambda/\mu_3}\right]\right.\\\left.\left. -\frac{\mu_3}{\mu_2^3} \left[1-e^{-\mu_2/\mu_3}\right]+\left[\frac{1}{\mu_2^2}- \frac{1}{\Lambda^2}\rule{0mm}{6mm}\right]\right). 
\right]
\end{multline}

As it can be seen Relation~\eqref{varPeq} gives an explicit, but intricate expression for the variance, we will present some numerical experiments based on this formula. The figures~\ref{fig1}, \ref{fig2} and~\ref{fig3} consider the case when the average number of proteins at equilibrium is fixed and equal to $300$, that $\lambda_2=0.02$, $\lambda_1^-=0.01$ and that the average of the lifetime of an mRNA [resp. protein] is $172$ [resp. $1000$]. We have considered several possible choices for the distribution $F_3$, it is assumed that all the other distributions  are exponential.  The parameter $S$ of the Gaussian is its variance.

\begin{figure}[ht]
\begin{center}
\scalebox{0.8}{\includegraphics{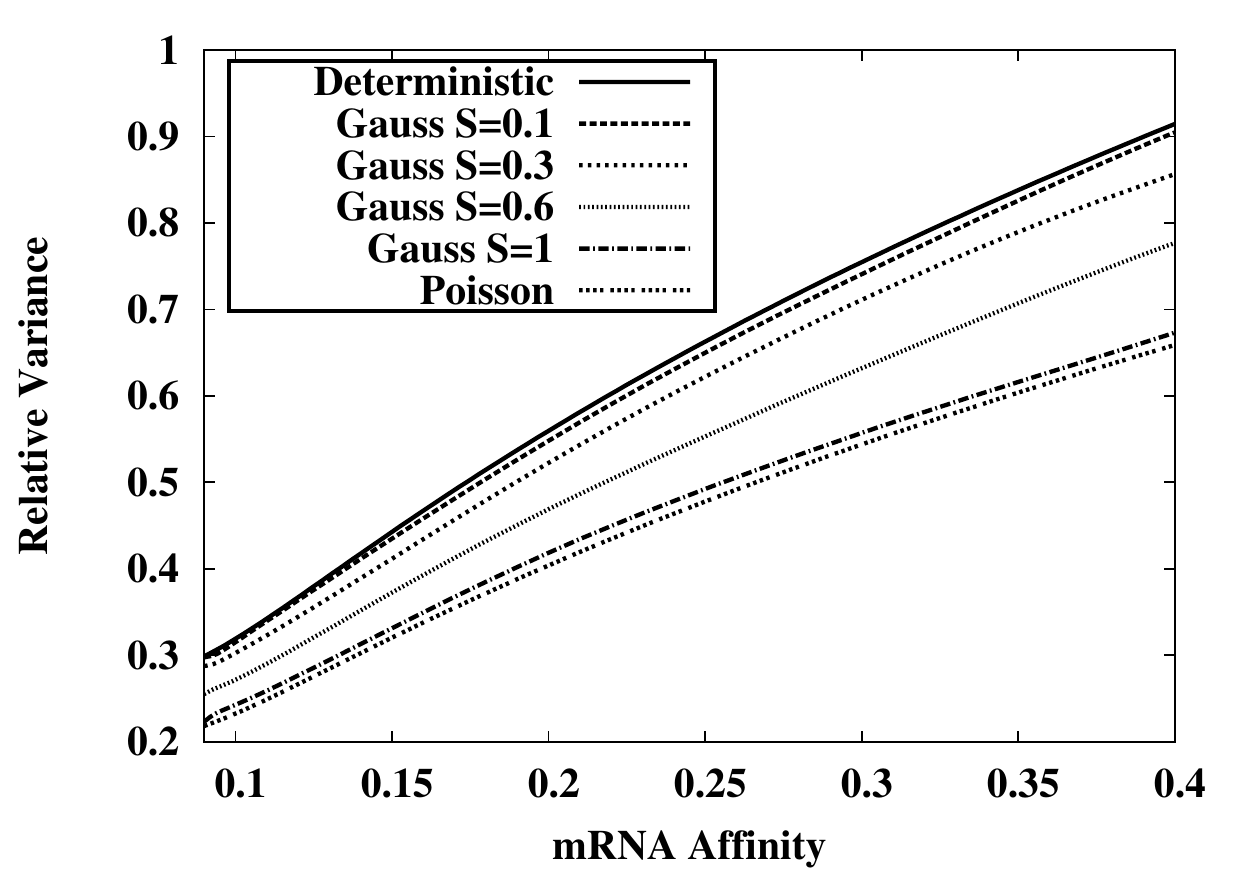}}
\caption{Square Root of Relative Variance of Nb of Proteins  with a fixed mean}\label{fig1}
\end{center}
\end{figure}

\begin{figure}[hbtp]
\begin{center}
\scalebox{0.8}{\includegraphics{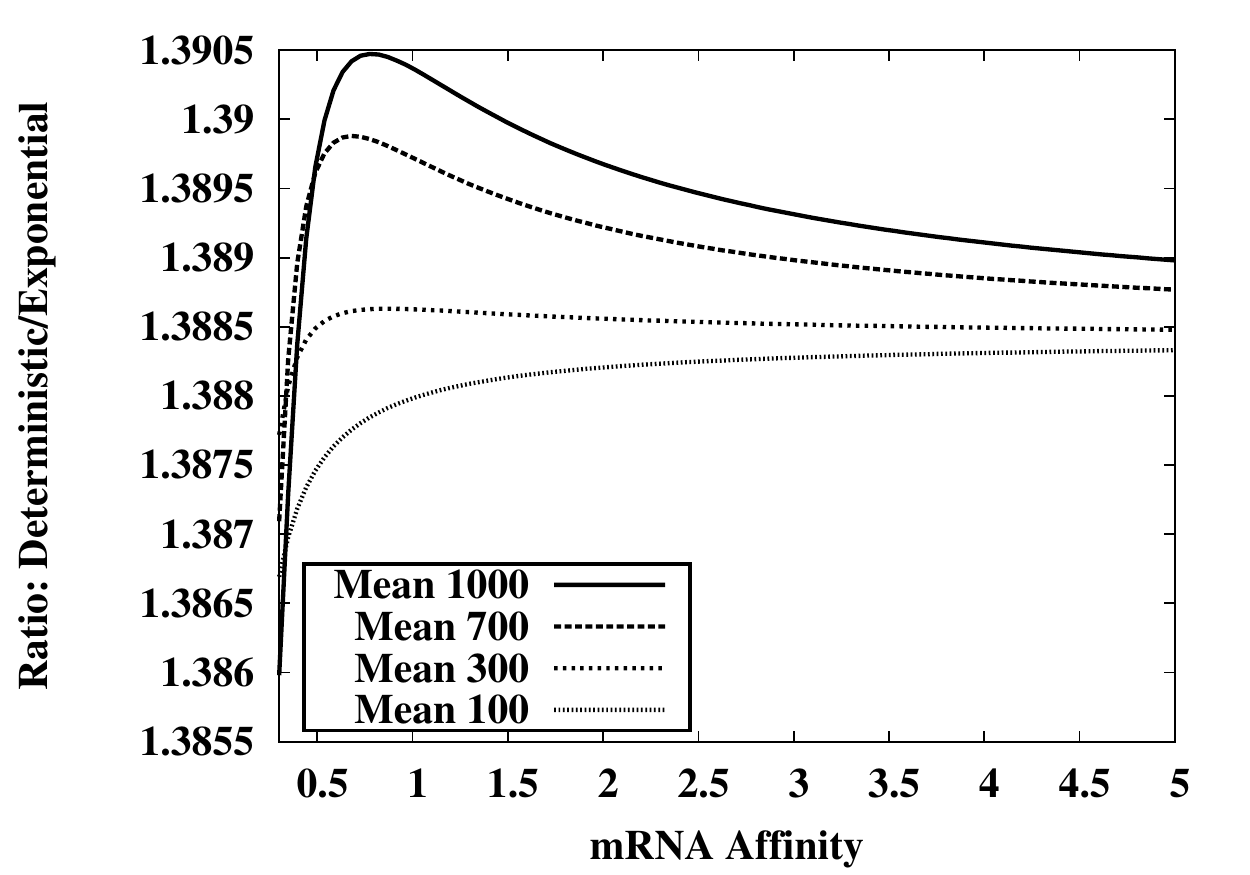}}
\caption{Square Root of Ratio of Variances of Nb of Proteins with a fixed mean}\label{fig2}
\end{center}
\end{figure}

\begin{figure}[hbtp]
\begin{center}
\scalebox{0.8}{\includegraphics{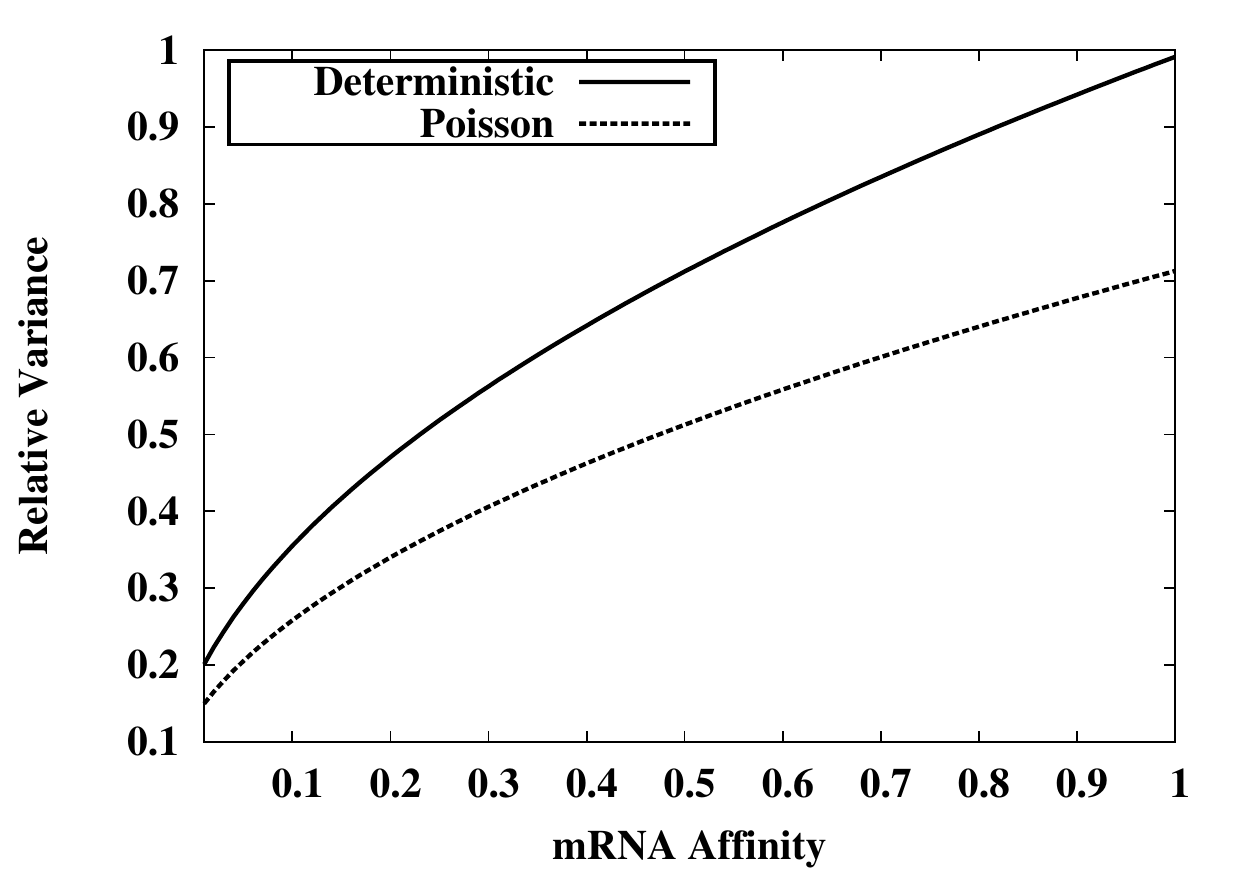}}
\end{center}
\caption{Square Root of Relative Variance of Nb of Proteins with a fixed mean}\label{fig3}
\end{figure}

\appendix

\section{A Reminder on Marked Poisson Processes}\label{secapp}
The main results concerning Poisson processes seen as marked point processes are briefly recalled. See Kingman~\cite{Kingman} and Chapter~1 of Robert~\cite{Robert} for a more detailed account. Throughout this section  $H$ is the space $\R^d$ for some $d\geq 1$. 
\begin{definition} 
 If $\lambda>0$, $\mu$ is a probability distribution on $H$, a marked Poisson process on $\R_+\times H$ with intensity $\lambda \diff x\otimes\mu$  is a sequence ${\cal N}_\lambda=(t_n,X_n)$  of elements of $\R_+\times H$ where
\begin{itemize}
\item $(t_n)$ is a (classical) Poisson process on $\R_+$ with rate $\lambda$. 
\item $(X_n)$ is an i.i.d. sequence with values in $H$ and whose distribution is $H$. 
\end{itemize}
\end{definition}
The sequence ${\cal N}_\lambda$ can also be seen as a marked point process on $\R_+\times H$, i.e. if $f:\R_+\times H\to\R_+$ is a continuous function then 
\[
{\cal N}_\lambda(f)=\int_{\R_+\times H} f(u,x){\cal N}_{\lambda}(\diff u,\diff x)=\sum_{n\geq 1} f(t_n,X_n). 
\]
In other words ${\cal N}_\lambda$ can also be seen as a sum of Dirac masses at the points $(t_n,X_n)$.  The following important proposition characterizes marked Poisson point processes. 
\begin{proposition}\label{remind}
The point process ${\cal N}_\lambda=(t_n,X_n)$ is a marked Poisson point process with intensity $\lambda \diff x\otimes\mu$ if and only if the relation
\begin{equation}\label{Lap}
\E\left(\exp\left(-{\cal N}_\lambda(f)\right)\right)=\exp\left(-\lambda\int_0^{+\infty} \left(1-e^{-f(u,x)}\right)\,\diff u\,\mu(\diff x)\right)
\end{equation}
holds for any non-negative continuous function $f$ on $\R_+\times H$.
\end{proposition}
The left-hand-side of Equation~\eqref{Lap} is usually defined as the Laplace transform of ${\cal N}_\lambda$ at $f$. This quantity determines completely the distribution of any marked point process. 

For $\xi>0$, by replacing $f$ by $\xi f$ in Relation~\eqref{Lap}, one gets an expression for 
\[
\E\left[\exp\left(-\xi{\cal N}_\lambda(f)\right)\right],
\]
if one differentiates it with respect to $\xi$ and sets $\xi=0$, the above identity gives
\begin{equation}\label{meanpp}
\E\left({\cal N}_\lambda(f)\right)=\E\left(\int_{\R_+\times H} f(u,x){\cal N}_{\lambda}(\diff u,\diff x)\right)
=\lambda \int_{\R_+\times H} f(u,x) \,\diff u\,\mu(\diff x).
\end{equation}

\providecommand{\bysame}{\leavevmode\hbox to3em{\hrulefill}\thinspace}
\providecommand{\MR}{\relax\ifhmode\unskip\space\fi MR }
\providecommand{\MRhref}[2]{%
  \href{http://www.ams.org/mathscinet-getitem?mr=#1}{#2}
}
\providecommand{\href}[2]{#2}

\end{document}